\documentclass[12pt,a4paper]{article}
\usepackage[latin1]{inputenc}
\usepackage{graphicx,amsmath,amsfonts,amssymb,mathrsfs,units}

\newcommand{\DATUM}{08-October-2010}         
\usepackage{bbm,MnSymbol}
\usepackage{amsthm,stmaryrd}
\usepackage{enumerate}

\newcommand{\Om}{\Omega}
 
\newcommand{\Lm}{\Lambda}

\newcommand{\ul}{\underline}

\newcommand{\cB}{\mathcal{B}}

\newcommand{\cE}{\mathcal{E}}
\newcommand{\cF}{\mathcal{F}}
\newcommand{\cG}{\mathcal{G}}

\newcommand{\cL}{\mathcal{L}}         %%%%%%%%%%%%%%%%%%%%%%%%
\newcommand{\cP}{\mathcal{P}}         %%%%%%%%%%%%%%%%%%%%%%%%

\newcommand{\cZ}{\mathcal{Z}}
\newcommand{\sH}{\mathscr{H}}

\newcommand{\RR}{\mathbb{R}}            %%%%%%%%%%%%%%%%%%%%%%%%%%%
\newcommand{\NN}{\mathbb{N}}            % Blackboard Bold Letters %
\newcommand{\CC}{\mathbb{C}}            %                         %
\newcommand{\PP}{\mathbb{P}}            %%%%%%%%%%%%%%%%%%%%%%%%%%%

\newcommand{\uA}{{\underline A}}
\newcommand{\uU}{{\underline U}}

\newcommand{\he}{\hat{e}}

\newcommand{\td}{\tilde{d}}

\newcommand{\tp}{\tilde{p}}

\newcommand{\tx}{\tilde{x}}

\newcommand{\tmu}{\tilde{\mu}}
\newcommand{\tnu}{\tilde{\nu}}

\newcommand{\tOm}{\widetilde{\Om}}
\newcommand{\tPP}{\widetilde{\mathbb{P}}}
\newcommand{\SU}{\mathrm{SU}}

\newcommand{\rRe}{\mathop{\mathrm{Re}}}  %%%%%%%%%%%%%%%%%%%%%%%%%%%%%%
\newcommand{\dist}{\mathop{\mathrm{dist}}}

\newcommand{\Tr}{\mathop{\mathrm{Tr}}}

\newcommand{\rk}{\mathop{\mathrm{rank}}}
\newcommand{\One}{\mathbbm{1}}
\newcommand{\dd}{\partial}
\newcommand{\inv}[1]{\frac{1}{#1}}
\newcommand{\wo}{\backslash}

\newcommand{\Z} {\mathbb{Z}}
\newcommand{\Zd}{\mathbb{Z}^d}

\newtheorem{Theo}{Theorem}[section]

\newtheorem{Lemma}[Theo]{Lemma}

\title{
The Integrated Density of States for the Wilson Dirac Operator}

\author{ 
V.~Bach\footnotemark[1], 
C.~Kurig\footnotemark[2], 
} 

\date{\DATUM}

\footnotetext[1]{v.bach@tu-bs.de}
\footnotetext[2]{kurig@mathematik.uni-mainz.de}

\begin{document}

\maketitle

\begin{abstract}
  It is shown that gauge field-dependent fermion Dirac operators from
  lattice QCD form an ergodic operator family in the probabilistic
  sense, provided the gauge field is an ergodic random field. As a
  consequence, the integrated density of states of such Dirac
  operators in the thermodynamic limit exists and is almost surely
  independent of the chosen gauge field configuration.
\end{abstract}

\newpage

%%%%%%%%%%%%%%%%%%%%%%%%%%%%%%%%%%%%%%%%%%%%%%%%%%%%%%%%%%%%%%%%%%%%%

%%%%%%%%%%%%%%%%%%%%%%%%%%%%%%%%%%%%%%%%%%%
\section{Introduction}%%%%%%%%%%%%%%%%%%%%%
%%%%%%%%%%%%%%%%%%%%%%%%%%%%%%%%%%%%%%%%%%%

The Standard Model of Elementary Particles is one of the biggest
achievements of theoretical physics during the second half of the
$20^{th}$ century. It provides a common conceptual basis for all
elementary forces except gravity. The part which describes the
strong nuclear force is called Quantum Chromodynamics (QCD). The
associated Lagrangian density has a clear and simple appearance 
that fits on a single line,
\begin{equation} \label{eq-1}
\cL(x) \ = \ 
- \frac{1}{4} \, F_{\mu \nu}(x) \, F^{\mu \nu}(x) 
\: + \: \overline{\psi}(x) \big( iD(\uA) - m \big) \psi(x),
\end{equation}
with $D(\uA) := \gamma_\mu(\dd_\mu + i A_\mu)$ being the Dirac operator of the
fermion field $\overline{\psi}(x)$, $\psi(x)$, which depends on the 
gauge field $\uA := (A_\mu(x))_{\mu,x}$, and 
$F_{\mu\nu}(x)=\dd_\mu A_\nu(x)-\dd_\nu A_\mu(x)-[A_\mu(x),A_\nu(x)]$ being
the field tensor. In spite of its structural simplicity, concrete quantitative
predictions are difficult to derive from \eqref{eq-1}, and one often resorts
either to calculations in perturbation theory or numerical
simulations on a discretized (Euclidean, after Wick rotation)
space-time, known as \emph{lattice QCD} (LQCD, see
\cite{MontvayMuenster1994} for an overview).

Several basic properties of QCD, such as the spontaneous breaking
of chiral symmetry or the phenomenon of quark confinement, manifest
themselves in the regime of low energies, where perturbation theory
in the QCD coupling constant cannot be applied. Spontaneous chiral
symmetry breaking is signaled by the formation of a non-vanishing
chiral condensate $\langle\overline\psi\psi\rangle$. In a seminal
paper \cite{BanksCasher1980}, Banks and Casher formulated a link between the value
of the condensate and the spectral properties of the Dirac operator
$D$ in the deep infra-red. Since the gauge field $\uA$ does not
appear explicitly in the observable, it acts as a background field
which nonetheless determines the spectral properties in a
non-trivial way.

The idea that the distribution of the low-lying eigenvalues of the 
fermion Dirac operator is very close to the one of the corresponding 
(i.e., respecting symmetries) random matrix ensemble was put forward 
in \cite{ShuryakVerbaarschot1993,VerbaarschotZahed1993,Verbaarschot1994}
and affirmed by numerous numerical studies, e.g. 
\cite{GiustiLuescher2003, BergMarkum2001}, 
for a review see for example \cite{VerbaarschotWettig2000}.
In fact, these
distributions agree to an accuracy that would, perhaps,
allow to replace the derivation of average spectral
properties of the fermion Dirac operator by sampling gauge field
configurations with the random matrix eigenvalue distribution.
The robustness of this phenomenon over a broad range of parameter
values, like the underlying gauge group or the system's temperature
(in the Boltzmann weight) is also remarkable and encourages us to
formulate our ultimate goal as to establish a rigorous mathematical
link between these distributions beyond numerical evidence.

The present paper is a first modest step towards this goal, namely to
fit the formulation of the model in LQCD into the mathematical
framework of \emph{ergodic operator families} which has been built up
over the past three decades or so to study random Schrödinger
operators and, especially, the Anderson model. In contrast to random
Schrödinger operators, however, the randomness in LQCD models lies
on the lattice bonds - not on the lattice sites - and corresponds
to a random magnetic field, rather than an alloy or a quenched glass.

As the main result of this paper we prove that Dirac operators of LQCD
which depend on the gauge field indeed constitute ergodic operator
families in the probabilistic sense, provided the gauge field itself
is ergodic, i.e., has sufficient rapidly decaying correlations (see
Section~\ref{sec:probmeas} for a precise formulation). This, in turn, is a
fair assumption in many physical situations, e.g., at high temperature
or if the gauge field is massive (which is believed to be true for
non-abelian gauge fields). As a consequence of our result the integrated 
density of states exists in the thermodynamic limit and is almost surely
independent of the chosen gauge field configuration.

%%%%%%%%%%%%%%%%%%%
Many observables can be expressed in terms of derivatives of the QCD
partition function with respect to source terms.  For example,
the chiral condensate is given by \cite{VerbaarschotWettig2000}
\begin{equation}
\langle \overline{\psi}{\psi} \rangle 
\ = \ -\lim_{m\to 0}\lim_{V\to\infty}\frac{1}{V}\dd_m \log Z^{QCD} ,
\end{equation}
with the partition function
\begin{equation}
Z^{QCD} \ = \ 
\int \prod_{x\in V, \mu=1,\ldots,d} dA_\mu \ 
\det\big[i D(\uA) + m \big] e^{-S_{YM}(\uA)} .
\end{equation}
Here, the integration of the fermionic variables yields the fermion 
determinant $\det[i D(\uA) + m]$, and $S_{YM}(\uA)$ is the Euclidean 
Yang-Mills action.

It is customary to use the \emph{quenched approximation} in numerical simulations, 
which amounts to setting the fermion determinant is equal to one.
This reduces the numerical effort significantly 
and corresponds to the physical case of infinitely heavy sea quarks.

The discretization of the Dirac operator is also subtle, because the
naive discretization leads to the occurrence of fermion doublers,
which have no physical meaning. There are several
ways to work around this problem. Wilson proposed to add a term that
vanishes in the continuum limit and suppresses the doublers on the
lattice \cite{Wilson1975}. Another method is to introduce staggered
fermions - the lattice is divided up in sub-lattices where different
staggered phases live, that are interpreted as physical phases
\cite{KogutSusskind1975, Susskind1977}. We are mainly interested in
those two cases, where the Dirac operator still has nearest-neighbour
interaction. This is not the case for another elegant solution, the
overlap operator proposed in \cite{Neuberger1997}.

%%%%%%%%%%%%%%%%%%%%%%%%%%%%%%%%%%%%%%%%%%%%%%%%%%%%%%%%%%%%%%%%%%%%%%%%%%%%%%%
\subsection{Introduction of the Model}\label{sec:introduction}%%%%%%%%%%%%%%%%%
%%%%%%%%%%%%%%%%%%%%%%%%%%%%%%%%%%%%%%%%%%%%%%%%%%%%%%%%%%%%%%%%%%%%%%%%%%%%%%%

Now we come to the precise description of the mathematical setting.
We consider the lattice $\Zd$, $d\geq 2$, with fermion fields supported on the sites and gauge fields supported on the bonds of the lattice. 
We take only the one particle case into account. Then the matter fields are complex vectors and the configuration of all matter fields is supposed to be an element of the Hilbert space $\sH= \ell^2(\Zd,\CC^k)$
of square summable $\Zd$-sequences in $\CC^k$. 
$\sH$ is equipped with the usual scalar product
\begin{equation}
	\langle \varphi, \psi \rangle \ =\ \sum_{x\in\Zd} \sum_{i=1}^d \bar{\varphi}_i(x) \psi_i(x),\quad \varphi,\psi \in \sH.
\end{equation}

We will consider operators on $\sH$, that also depend on the configuration of gauge fields on the bonds of $\Zd$.
The set of bonds in the lattice $\Zd$ is denoted by
\begin{equation}
	\cB \ :=\ \Zd\times \{1,\ldots,d\}.
\end{equation}
The bond $(x,\mu)\in\cB$ is the one connecting $x$ and $x+\he_\mu$, with $\he_\mu$ the unit vector in $\Zd$ pointing in direction $\mu$. We give the bond $(x,\mu)$ the orientation from $x$ to $x+\he_\mu$.

The gauge fields associated to the bonds are elements of a compact Lie group $\cG$, the gauge group. We assume $\cG$ to be either $SO(N)$, $SU(N)$, or $U(N)$ to be explicit and since these are the relevant physical cases.
The gauge field on the bond $b=(x,\mu)$ is denoted by $U_{x,\mu}$ or $U_b$. 

It turns out to be necessary to consider the direction of a bond, the gauge field for going from $x+\he_\mu$ to $x$ is $U_{x+\he_\mu,-\mu}$ and we set $U_{x+\he_\mu,-\mu}=U_{x,\mu}{}^{-1}$.

A gauge field configuration is the collection $\{ U_b \}_{b\in\cB}$.
As mentioned before, this gauge field configuration is randomly generated.
To specify the underlying probability space we will need the notion of a plaquette, a collection of four bonds that form a plane square in $\Zd$,
\begin{equation}
	p(x;\mu,\nu) \ :=\ \big\{ (x,\mu), (x+\he_\mu,\nu), (x+\he_\nu,\mu), (x,\nu) \big\} ,
\end{equation}
with $\mu\neq\nu$.
We need the product of the gauge fields along a plaquette $p=p(x;\mu,\nu)$,
\begin{equation}
	U_p \ :=\ U_{x;\mu,\nu}\ :=\ 
	U_{x,\nu}^{-1}\  U_{x+\he_\nu,\mu}^{-1}\  U_{x+\he_\mu,\nu}\ U_{x,\mu},
\label{eq:defUp}
\end{equation}
where the orientation of the bonds leads to the inverse gauge fields.
Thus we have $U_{x;\mu,\nu}=U_{x;\nu,\mu}^{-1}$ and define a plaquette as positively orientated if $\mu<\nu$.

The set of all positively orientated plaquettes is denoted by
\begin{equation}
\cP \ :=\ \big\{  p(x;\mu,\nu) \ \big|\ x \in \Zd,\ \mu,\nu \in \{1,\ldots,d\big\},\ \mu<\nu  \}.
\end{equation}

%%%%%%%%%%%%%%%%%%%%%%%%%%%%%%%%%%%%%%%%%%%%%%%%%%%%%%%%%%%%%%%%%%
\subsubsection{The Probability Space}\label{sec:probspace}%%%%%%%%
%%%%%%%%%%%%%%%%%%%%%%%%%%%%%%%%%%%%%%%%%%%%%%%%%%%%%%%%%%%%%%%%%%

We start by specifying the probability space for a single bond. 
Since $\cG$ is a compact Lie group, $\cG$ is equipped with a natural measure, the Haar measure $\mu_H$.
Because of the compactness of $\cG$, the Haar measure $\mu_H$ is normalized and we obtain the probability space $(\cG,\cF_1,\mu_H)$, with $\cF_1$ being the $\sigma$-algebra of Borel sets of the topological space $\cG$.
A priori, we let the gauge field $U_b$ on a single bond be a random variable that is uniformly distributed with respect to the Haar measure $\mu_H$ of $\cG$. 

The gauge field configuration $\uU=\{ U_b \}_{b\in\cB}$ can be regarded as an element of the product space $\cG^\cB$.
$\cG^\cB$ is compact in the product topology by Tychonov's Theorem.
The probability space for gauge field configurations is now constructed by means of cylinder sets as in \cite{Billingsley1979}.
A cylinder set is subset of $\cG^\cB$ of the form
\begin{equation}
	M\ =\ \big\{ \uU \in \cG^\cB \ \big|\ U_{b_1} \in A_1, \ldots, U_{b_n} \in A_n \big\}
\end{equation}
with $A_1,\ldots,A_n \in \cF_1$ and $b_1,\ldots, b_n \in \cB$. The set of all cylinder sets is denoted $\cZ$.
We take as a $\sigma$-algebra for $\cG^\cB$ the $\sigma$-algebra $\cF$ generated by the system of all cylinder sets.
The probability-measure $\tPP$ on $\cG^\cB$ is then defined to be the product measure, setting for every cylinder set
\begin{equation}
	\tPP\big(\big\{ \uU \in \cG^\cB \ \big|\ U_{b_1} \in A_1, \ldots, U_{b_n} \in A_n \big\}\big)
	\ =\ \prod_{i=1}^{n} \mu_H( A_i ).
\end{equation}
Now, we modify the measure $\tPP$ by a weight function that represents the gauge action. 
Formally, the measure $\PP$ is defined as
\begin{equation}
 d\PP(\uU) \ =\ Z^{-1} e^{-S(\uU)}\ d\tPP(\uU)
\end{equation}
with $Z$ being a normalization factor.
We assume the gauge action to be of the following form:
\begin{equation}
	S(\uU)\ =\ \beta \sum_{p\in\cP} \rRe \Tr(\One - U_p)
\end{equation}
with $\beta>0$  and $U_p$ the plaquette variable as defined in (\ref{eq:defUp}).
This action on discrete space-time is the Wilson action used in lattice QCD calculations.
Note that the gauge action is invariant under translations. 
Denoting by $T^\ell$ the translation in $\Zd$ by $\ell \in \Zd$, i.e. $	T^\ell x := x -\ell$, and defining the translation of a gauge field by $\ell \in \Zd$ to be 
\begin{equation} \label{eq:deftrans0}
	T^\ell U_{x,\mu}\ :=\ U_{T^\ell x,\mu}\ =\ U_{x-\ell,\mu},
\end{equation}
for any $U_{x,\mu} \in \cG$ and $(x,\mu) \in \cB$, we have 
\begin{equation}
	S(\uU)\ =\ S(T^\ell\uU)
\end{equation}
for all gauge field configurations $\uU\in\cG^\cB$ and all $\ell \in \Z^d$.

In the following, we use the Gibbs formalism to establish the existence of $\PP$ and its uniqueness for small $\beta$, following \cite{Bovier2006}.
To this end, we fix a finite subset $\Lm$ of $\cB$ and a gauge field configuration $\eta \in \cG^\cB$ that represents the boundary condition, which fixes the gauge field outside of $\Lm$.
The corresponding local specifications are
\begin{equation}
 	d\mu_{\Lm,\beta}^{\eta}(\uU)\ =\ 
	\big(Z_{\Lm,\beta}^{\eta}\big)^{-1} 
	\exp \bigg[- \beta \sum_{p \cap \Lm \neq \emptyset} \rRe \Tr(\One-U_p )\bigg]\ 
	\prod_{b\in\Lm} d\mu_H(U_b)
\end{equation}
where
\begin{equation}
	Z_{\Lm,\beta}^{\eta}\ =\ 
	\int_{\cG^\Lm} \exp\bigg[- \beta \sum_{p \cap \Lm \neq \emptyset} \rRe\Tr(\One-U_p)\bigg]\ 
	\prod_{b\in\Lm} d\mu_H(U_b).
\end{equation}
It is known (see for example \cite{Simon1993, Bovier2006}) that, for sufficiently small $\beta>0$, these measures weakly converge, as $\Lm\nearrow\cB$, to a unique Gibbs measure, since $\cG^\cB$ is compact, the action is continuous in the gauge fields, and for all $b \in \cB$ there is a constant $c' < \infty$, such that
\begin{equation}
	 \sum_{p\in\cP: b \in p} \left\|  \rRe \Tr(\One - U_p) \right\|_\infty \ \leq\ c'.
\end{equation}
Furthermore, the Gibbs measure is independent of the boundary condition $\eta$.
In our case, we have 
\begin{equation}
	 \sum_{p\in\cP: b \in p} \left\|  \rRe \Tr(\One - U_p) \right\|_\infty \ \leq\ 2(d-1)\cdot 2N,
\end{equation}
uniformly in $b\in\cB$, since any bond $b\in\cB$ is element of $2(d-1)$ plaquettes.

Dobrushin's uniqueness criterion ensures the uniqueness of the Gibbs measure for translation invariant interactions, provided
\begin{equation}
	0\ < \sum_{p \in \cP: b_0 \in p} (|p|-1)  \left\| \rRe \Tr(\One - U_p) \right\|_{\infty}\ <\ \beta^{-1},
\end{equation}
for one bond (and, hence, for all bonds) $b_0\in\cB$, see for example \cite{Simon1993}.
In our case, the criterion is fulfilled for all
\begin{equation}
  0\ <\ \beta\ < \frac{1}{12 N (d-1)}.
\label{eq:critbeta}
\end{equation}
Thus the Gibbs measure $\PP$ exists and is unique, for sufficiently small $\beta>0$.

We note in passing, that the translations are measure preserving transformations with respect to $\tPP$ and $\PP$, that means for all $A \in \cF$, $\ell \in \Z^d$
\begin{equation}\label{eq:transinv}
	 \tPP(T^\ell A)=\tPP(A)  \quad \text{and} \quad \PP(T^\ell A)= \PP(A).
\end{equation}
Put differently, $\PP$ and $\tPP$ are stationary w.r.t. the group $\Zd$ of translations.

%%%%%%%%%%%%%%%%%%%%%%%%%%%%%%%%%%%%%%%%%%%%%%%%%%%%%%%%%%
\subsubsection{Ergodic Probability Measures}\label{sec:probmeas}%%%%%%%%%%%%%%
%%%%%%%%%%%%%%%%%%%%%%%%%%%%%%%%%%%%%%%%%%%%%%%%%%%%%%%%%%

A stationary probability measure $P$ is called \emph{ergodic} iff, for all $A,A' \in \cF$,
\begin{equation} \label{eq:deferg}
 	\frac{1}{(2L+1)^d} \sum_{l\in \Zd,\ \|l\|_{\infty}\leq L} P(A \cap T^l A') \to P(A)P(A'),\quad \text{as} L\to \infty ,
\end{equation}
with $\|l\|_{\infty}=\max\big\{ |l_1|,\ldots,|l_d| \big\}$.

A random variable $f:(\cG^\cB,\cF)\to(\RR,B)$ is called \emph{invariant} iff $f(T^\ell \uU)=f(\uU)$, for all $\ell \in \Zd$ and almost all $\uU \in \cG^\cB$. The importance of the notion of ergodicity lies in the fact that any invariant random variable is $P$-almost surely constant.

In the following, we show that the measure $\PP$ is ergodic, provided (\ref{eq:critbeta}) holds true.
First, we specify the decay of correlations. 
We define a metric $d: \cB \times \cB \to \RR^+_0$ on $\cB$ by setting $d(b,b):=0$ and
\begin{eqnarray}	
	\lefteqn{
	d(b,b')\ :=
	}\\\nonumber&&\hspace{-5mm}
	\min\big\{ n \ |\ \exists \{p_1,\ldots,p_n\} \subset \cP :\ 
	b\in p_1,\ p_1 \cap p_2 \neq \emptyset,\ldots, p_{n-1} \cap p_n \neq \emptyset,\ b' \in p_n  \big\},
\end{eqnarray}
for $b,b'\in\cB$, $b\neq b'$. I.e. $d(b,b')$ is the minimal number of plaquettes to connect $b$ and $b'$. 

If $p=(x;\mu,\nu)$, $\tp=(\tx;\tmu,\tnu)$, and $p\cap\tp\neq\emptyset$ then $\|x-\tx\|_\infty\leq 1$.
Therefore, the minimal number of plaquettes connecting $b=(x,\mu) \in \cB$ and $b'=(y,\nu)\in\cB$ is at least $\|x-y\|_\infty$.
Observing that, for all $x\in\Z^d$, $\mu,\nu,\tau\in\{1,\ldots,d\}$, $\nu\neq\mu\neq\tau$,
\begin{eqnarray}
	p(x;\mu,\nu) \cap p(x+\he_\mu;\mu,\nu) &\neq& \emptyset \nonumber\\
	p(x;\mu,\nu) \cap p(x-\he_\mu;\mu,\nu) &\neq& \emptyset \nonumber\\
	p(x;\mu,\nu) \cap p(x;\mu,\tau) &\neq& \emptyset	,
\end{eqnarray}
we obtain, with $b=(x,\mu)$, $b'=(y,\nu)\in\cB$ as above, that
\begin{equation}
	\|x-y\|_\infty\ \leq\ d(b,b')\ \leq\ \|x-y\|_1 + d,
\end{equation}
with $\|x-y\|_1=\sum_{i=1}^d |x_i-y_i|$.\\
A modified version of the metric $d$ called $\td:=\ln(\frac{c}{\beta})d$ is obtained by multiplying $d$ with $\ln(\frac{c}{\beta})$, where $c>\beta$.

Now, we take two cylinder sets $A,A'\in\cF$ and show that condition (\ref{eq:deferg}) is fulfilled for $A$ and $A'$. 
Since the sigma-algebra $\cF$ is generated by $\cZ$, (\ref{eq:deferg}) extends to all $\cF$.
There are two finite sets $\Lm_A,\Lm_{A'}\subset\cB$ and $A_b,A'_{b'}\in \cF_1$ for all $b\in\Lm_A$, $b'\in\Lm_{A'}$ such that 
\begin{equation}
	A=\bigtimes_{b\in\Lm_A} A_b \times \cG^{\cB\wo\Lm_A},\ \ 	A'=\bigtimes_{b'\in\Lm_{A'}} A'_{b'} \times \cG^{\cB\wo\Lm_{A'}}
\end{equation}
Now choose $\chi_{A,\epsilon}, \chi_{A',\epsilon} \in C(\cG^\cB,[0,1])$ such that $\chi_{A,\epsilon}, \chi_{A',\epsilon}$ depend only on the variables $U_b$ with $b \in\Lm_A$ or $b\in\Lm_{A'}$, respectively, $\chi_{A,\epsilon}(\uU)=1$, $\chi_{A',\epsilon}(\uU')=1$, for all $\uU \in A$, $\uU' \in A'$, and 
\begin{equation}
	\PP\big\{ \uU \in A^C \ | \ \chi_{A,\epsilon}(\uU)>0 \big\} < \epsilon\ , \quad	\PP\big\{ \uU' \in A'^C \ | \ \chi_{A',\epsilon}(\uU')>0 \big\} < \epsilon.
\end{equation}
Then $\chi_{A,\epsilon}, \chi_{A',\epsilon}$ are continuous functions that differ from the characteristic function of $A$, $A'$ only on an set of measure less than $\epsilon$.

Now, we use a result of \cite{Foellmer1982} summarized in \cite{Simon1993}. It states that, if $\gamma<1$, where $\gamma$ is a constant depending on the interaction, one gets for any two bonds $i,j\in\cB$ that
\begin{eqnarray}
	 \lefteqn{\left|\int_{\cG^\cB} \chi_{A,\epsilon}\ \chi_{A',\epsilon}\ d\PP - 
	 \int_{\cG^\cB} \chi_{A,\epsilon}\ d\PP\ \int_{\cG^\cB} \chi_{A',\epsilon}\ d\PP\ \right| }\nonumber\\
	 &	\leq & \frac{1}{4}\ e^{-\td(i,j)} (1-\gamma)^{-1}\ \Delta_{i}(\chi_{A,\epsilon})\ \Delta_{j}(\chi_{A',\epsilon}),
\label{eq:expdecay1}
\end{eqnarray}
where
\begin{equation}
  \Delta_{j}(f):=\sum_{i\in\cB} e^{\td(i,j)} \sup\big\{\ |f(\uU)-f(\uU')| \ \big|\ U_b=U'_b,\ b \neq i\ \big\}
\end{equation}
and 
\begin{equation}
	\gamma=\sup_j \sum_{i\in\cB, i\neq j} e^{\td(i,j)} \rho_{ij}.
\end{equation}
In our case $\rho_{ij}$ for $i\neq j$, can be estimated as
\begin{equation}
	\rho_{ij} \leq \sum_{p\in\cP: {i,j}\in p} \left\| \beta \rRe \Tr(\One-U_p) \right\| \leq  2N\beta\ \One[d(i,j)=1] ,
\end{equation}
such that we get $\gamma \leq 3 \cdot 2(d-1) 2N\beta \cdot \frac{c}{\beta}$ ,and $\gamma<1$ corresponds to
\begin{equation}
	c \ <\ \frac{1}{12N(d-1)},
\end{equation}
arriving at Condition (\ref{eq:critbeta}).
We denote the \emph{distance} of $A$ and $A'$ by
\begin{equation}
	\dist(A,A'):= \min\{d(i,j) | i \in \Lm_A, j \in \Lm_{A'} \}
\end{equation}
and the \emph{diameter} of $A$ by
\begin{equation}
	D_A := \max \{ d(i,j) | i,j \in \Lm_A \}.
\end{equation}
If $i\in\Lm_A$ then
\begin{equation}
	\Delta_i(\chi_{A,\epsilon}) \ \leq \ \sum_{k\in\cB}  e^{\td(i,k)} \One[k\in\Lm_A] \ \leq \ |\Lm_A| \big(\frac{c}{\beta} \big)^{D_A}
\end{equation}
and analogously $	\Delta_j(\chi_{A',\epsilon}) \leq |\Lm_{A'}| (\frac{c}{\beta} )^{D_{A'}}$, provided $j\in\Lm_{A'}$.
Inserting this into Equation (\ref{eq:expdecay1}), and choosing $i\in\Lm_A$ and $j \in\Lm_{A'}$ such that $d(i,j)=\dist(A,A')$, we estimate
\begin{eqnarray}
	\lefteqn{\left|\int_{\cG^\cB} \chi_{A,\epsilon}\ \chi_{A',\epsilon}\ d\PP - 
	 \int_{\cG^\cB} \chi_{A,\epsilon}\ d\PP\ \int_{\cG^\cB} \chi_{A',\epsilon}\ d\PP\ \right| }\nonumber\\
	&\leq & \frac{1}{4} \bigg(\frac{\beta}{c}\bigg)^{\dist(A,A')-(D_A+D_{A'})} \frac{|\Lm_A|\ |\Lm_{A'}|}{1-12N(d-1)c},
\end{eqnarray}
for all $0<\beta<c<(12N(d-1))^{-1}$. In the limit $\epsilon\to 0$, we obtain
\begin{equation}
	\left| \PP(A\cup A') - \PP(A)\PP(A') \right|\ \leq\ C_{A,A'} \bigg(\frac{\beta}{c}\bigg)^{\dist(A,A')} ,
\end{equation}
where
\begin{equation}
	C_{A,A'}\ =\ \frac{1}{4}\ \frac{1}{1-12N(d-1)c} |\Lm_A| |\Lm_{A'}| \bigg(\frac{c}{\beta}\bigg)^{D_A+D_{A'}}
\end{equation}
is independent of the distance of $A$ and $A'$.

The exponential decay of correlations implies at once,
\begin{eqnarray}
	\lefteqn{	\frac{1}{(2L+1)^d} \sum_{\ell\in \Zd,\ \|\ell\|_{\infty}\leq L} \left|\ \PP(A \cap T^\ell A') - \PP(A)\ \PP(A')\ \right|}		\nonumber\\
	&\leq& \frac{1}{(2L+1)^d} \sum_{\ell\in \Zd,\ \|\ell\|_{\infty}\leq L} C_{A,A'} \bigg(\frac{\beta}{c}\bigg)^{\dist(A,T^\ell A')}	\nonumber\\
	&\leq& \frac{1}{(2L+1)^d} \sum_{m=0}^{L} 2d (2m+1)^{d-1} C_{A,A'} \bigg(\frac{\beta}{c}\bigg)^{m-\dist(A,A')-2(D_A + D_{A'})}	\nonumber\\
	&\leq&	\frac{2dC_{A,A'}}{2L+1} \bigg(\frac{\beta}{c}\bigg)^{-\dist(A,A')-2(D_A+D_{A'})} \frac{1}{1-\nicefrac{\beta}{c}} \stackrel{L\to\infty}{\longrightarrow} 0  ,
\end{eqnarray}
where we used that $d(k,T^\ell k)\leq d(i,k) + d(i,j) + d(T^\ell k,j)$ and thus
\begin{eqnarray}
	\dist(A,T^\ell A')&= &\min\{d(i,j) | i \in \Lm_A, j \in T^\ell \Lm_{A'} \}  \nonumber\\
	&\geq& \min\{ d(k,T^\ell k) | k \in \Lm_{A'} \} 
			 - \max\{ d(i,k) | i \in \Lm_A, k \in \Lm_{A'} \} \nonumber\\
		&&	 - \max\{ d(T^\ell k,j) | k \in \Lm_{A'} ,  j \in T^\ell \Lm_{A'} \}\nonumber\\
	&\geq& \|\ell\|_\infty - D_A - 2 D_{A'} - \dist(A,A').
\end{eqnarray}

The validity of (\ref{eq:deferg}) on cylinder sets extends to the sigma-algebra $\cF$ generated by the cylinder sets $\cZ$ by a monotone class argument.
Thus Condition (\ref{eq:critbeta}) for the uniqueness of $\PP$ ensures its ergodicity, too.

%%%%%%%%%%%%%%%%%%%%%%%%%%%%%%%%%%%%%%%%%%%%%%%%%%%%%%%%%%%%%%%%%%%%%%%%%%%
\subsubsection{Ergodic Families of Wilson Dirac Operators}%%%%%%%%%%%%%%%%%
%%%%%%%%%%%%%%%%%%%%%%%%%%%%%%%%%%%%%%%%%%%%%%%%%%%%%%%%%%%%%%%%%%%%%%%%%%%
We specify the considered operators. The dependence of those operators $D_\uU$ on the gauge field configuration is emphasized by the index $\uU$. We consider the corresponding family of operators $\{D_\uU\}_{\uU \in \cG^\cB}$.

Let $\{D_\uU\}_{\uU \in \cG^\cB}$ be a family of bounded, self-adjoint operators on $\sH=\ell^2(\Z^d;\CC^k)$. We call this family \emph{stationary} if it depends on the gauge field configuration $\uU \in \cG^\cB$ in such a way that translations act transitively, i.e.,
\begin{equation} \label{eq:defop}
 \tau^\ell D_\uU \tau^{-\ell} = D_{T^\ell \uU}
\end{equation}
for all $\uU \in \cG^\cB$ and $\ell \in \Zd$, where $\tau^\ell$ denotes the corresponding translation on $\sH$, i.e.,
\begin{equation} \label{eq:deftrans1}
	[\tau^\ell \phi](x)= \phi(x-\ell),
\end{equation}
for any $\phi \in \sH$, $x\in \Zd$.
A stationary family $\{D_\uU\}_{\uU \in \cG^\cB}$ is called \emph{ergodic} if the underlying probability measure $\PP$ on $\cG^\cB$ is stationary and ergodic.
The crucial fact about ergodic families $\{D_\uU\}_{\uU \in \cG^\cB}$ is the independence of their spectra on $\uU$, $\PP$-almost surely. \cite{Pastur1980}.

We assume that $D_\uU$ includes only nearest-neighbour interaction, i.e., for $\phi \in \sH,\ x \in \Z^d$, $[D_\uU \phi](x)$ depends only on the values of $\phi(x)$, $\phi(y)$ for those $y$ with $|x-y|=1$ and the gauge fields $U_{x,\mu}$, $U_{x,-\mu}$ for $\mu \in \{1,\ldots,d\}$, where we use the notation $U_{x,-\mu}:=U_{x-\he_\mu,\mu}$.

There are various examples for such operators of physical interest. As mentioned in the introduction, we are mainly interested in the Wilson Dirac operator and the staggered fermions operator. For simplicity, we concentrate our attention to the Wilson Dirac operator in this paper. 

In lattice gauge theories the Wilson Dirac operator $D$ is used \cite{MontvayMuenster1994}, which is a discretized version of the QCD-Dirac operator $D_{QCD}= \gamma_\mu (\dd_\mu + i A_\mu) + m $ with gauge fields $A_\mu$. 
The corresponding matter fields are defined on the hypercubic lattice $\Z^4$ and are assumed to have a Dirac structure labeled by Dirac indices $\alpha\in\{1,2,3,4\}$, as well as a colour structure with labels $c\in\{1,\ldots,N_c\}$.
The Dirac structure is represented by the $4\times4$ Euclidean Dirac matrices $\{\gamma_\mu\}_{\mu=1,\ldots,4}$. A customary explicit representation is
\begin{equation}
	\gamma_{1,2,3}= \bigg(
\begin{array}{c c}
	0 & -i\sigma_{1,2,3}\\
	i \sigma_{1,2,3} & 0 \\
\end{array}	\bigg)
\quad ,\quad
	\gamma_4= \bigg(
\begin{array}{c c}
	0 & \One\\
	\One & 0 \\
\end{array}	\bigg)
\end{equation}
with $\sigma_{1,2,3}$ being the Pauli matrices.
The Dirac matrices form a Clifford-Algebra since they fulfill $\{\gamma_\mu,\gamma_\nu\}=2\delta_{\mu \nu}$.
Introducing $\gamma_5 := \gamma_1 \gamma_2 \gamma_3 \gamma_4$, i.e.,
\begin{equation}
 \gamma_5=\bigg(
\begin{array}{c c}
	\One & 0\\
	0 & -\One \\
\end{array}	\bigg),
\end{equation}
we observe that $\{ \gamma_\mu, \gamma_5\} = 0$.

The gauge group is $\cG=SU(N_c)$ and acting on the colour structure.
Therefore, $k=4\cdot N_c$, and
$ \phi = \{\phi_{\alpha,c}\}_{ \alpha=1,\ldots,4 , \ c=1\ldots,N_c } \in \ell^2(\Z^4,\CC^k)$.
The Wilson Dirac operator $D$ is defined by
\begin{eqnarray} \label{eq:defDirac}
	\lefteqn{[D \phi]_{\alpha,c}(x) 
	\ :=\ \sum_{\beta=1}^{4} \bigg\{ (\gamma_5)_{\alpha,\beta} \phi_{\beta,c}(x) }\nonumber\\
	&&-\kappa \sum_{\mu=1}^{4} \sum_{\sigma=\pm1} 
		 \sum_{f=1}^{N_c} \big( (r \gamma_5)_{\alpha,\beta} - \sigma (\gamma_5 \gamma_\mu)_{\alpha,\beta}\big)\ (U_{x,\sigma\mu})_{c,f}\ \phi_{\beta,f}(x+\sigma\he_\mu) \bigg\} ,
\end{eqnarray}
in short,
\begin{equation} \label{eq:defDirac2}
	[D \phi](x) = \gamma_5\big[ \phi(x) - \kappa \sum_{\mu=1}^{4} \sum_{\sigma=\pm1} (r-\sigma\gamma_\mu) U_{x,\sigma\mu}\ \phi(x+\sigma\he_\mu)\big]\ . 
\end{equation}
The parameter $r\in (0,1]$ is the Wilson parameter and $\kappa>0$ the hopping parameter. 

Displaying the dependence of $D$ on the gauge field configuration $\uU$ by writing $D_\uU$, we observe that $D_\uU$ fulfills condition~(\ref{eq:defop}) for any $\uU \in \SU(N_c)^\cB$ and any $\phi \in \sH, x \in \Zd$,
\begin{align}
	[\tau^\ell D_\uU \tau^{-\ell} \phi](x)
	&= [D_\uU \tau^{-\ell} \phi](x-\ell) \nonumber \\
	&= [\tau^{-\ell} \phi](x-\ell) - \kappa \sum_{\mu=1}^{4} \sum_{\sigma=\pm1} 
		 (r+\gamma_{\sigma\mu}) U_{x-\ell,\sigma\mu}\ [\tau^{-\ell} \phi](x-\ell+\sigma\he_\mu) \nonumber \\	
	&= \phi(x) -\kappa \sum_{\mu=1}^{4} \sum_{\sigma=\pm1} (r+\gamma_{\sigma\mu}) U_{x-\ell,\sigma\mu}\ \phi(x+\sigma\he_\mu) \nonumber \\	
	&= [D_{T^\ell \uU} \phi](x)
\end{align}
and hence we have 
\begin{equation}
	\tau^\ell D_\uU \tau^{-\ell} = D_{T^\ell \uU},
\end{equation}
for all $\ell \in \Zd$ and $\uU \in SU(N_c)^\cB$.
Thus, if $\PP$ is ergodic, so is $\{ D_\uU \}_{\uU \in SU(N_c)^\cB}$, and its spectrum is $\PP$-almost surely constant.

%%%%%%%%%%%%%%%%%%%%%%%%%%%%%%%%%%%%%%%%%%%%%%%%%%%%%%%%%%%%%%%%%%%%%%%%
\subsection{Introduction of the Integrated Density of States}%%%%%%%%%%%
%%%%%%%%%%%%%%%%%%%%%%%%%%%%%%%%%%%%%%%%%%%%%%%%%%%%%%%%%%%%%%%%%%%%%%%%
In the following we study the \emph{integrated density of states} of $\{ D_\uU \}_{\uU \in \cG^\cB}$, which represents the number of eigenstates per unit volume. For the precise definition of the integrated density of states, we restrict our analysis to a finite subset $\Lm \subset \Z^d$. 
Besides, this also allows us to relate our analysis to numerical simulation.

The \emph{boundary} of $\Lm$, denoted $\dd\Lm$, is defined as
\begin{equation}
	 \dd \Lm\ :=\ \big\{ y \in \Lm \ | \ \exists\ x \in \Lm^C,\ |x-y|=1\big\}\ \subseteq\ \Lm.
\end{equation}
Furthermore we use the canonical orthonormal basis $\cE_\sH:=\{ \hat{s}^{(x,i)}\}_{x \in \Zd, i \in \{1,\ldots,k\} }$ of $\sH$ where the $\Zd$-sequence $\hat{s}^{(x,i)}$ is set to be
\begin{equation}\label{def:ONB}
	\hat{s}^{(x,i)}(y):=
	\begin{cases}
		\ \hat{e}_i , &  x = y \\	\ 0,	&  x \neq y 
	\end{cases}
\end{equation}
with $\hat{e}_i \in \CC^k$ the unit vector in  direction $i$.

Since $D_\uU$ contains only nearest-neighbour hopping, the value of $D_\uU \hat{s}^{(x,i)}$ does not change, if we replace $D_\uU$ by a restriction of $D_\uU$ to $\Lm$, for any point $x$ in $\Lm \wo \dd\Lm$. Only the boundary $\dd\Lm$ needs further specification. We present two customary choices for this, namely, Dirichlet boundary conditions and periodic boundary conditions.

%%%%%%%%%%%%%%%%%%%%%%%%%%%%%%%%%%%%%%%%%%%%%%%%%%%%%%%
\subsubsection{Dirichlet Boundary Conditions}%%%%%%%%%%
%%%%%%%%%%%%%%%%%%%%%%%%%%%%%%%%%%%%%%%%%%%%%%%%%%%%%%%
First, we restrict the operators to the finite subset $\Lm \subset \Z^d$ by means of the projection
\begin{equation}
	P^{(dir)}_\Lm: \sH \to \sH^{(dir)}_\Lm,\quad 
	\big[ P^{(dir)}_\Lm \varphi \big](x):=
	\begin{cases} \	\varphi(x), & x\in \Lm, \\	\ 0, 	& x \notin \Lm, \end{cases} 
\end{equation}
with 
\begin{equation}
	\sH^{(dir)}_\Lm\ =\ \ell^2(\Lm,\CC^k) \subset \sH
\end{equation}
being the Hilbert space of sequences vanishing outside $\Lm$.
Note that
\begin{equation}
	P_{T^\ell \Lm} = \tau^\ell P_\Lm \tau^{-\ell} .
\end{equation}
Then we define 
\begin{equation}
D^{(dir)}_{\uU,\Lm} = P^{(dir)}_\Lm D_\uU P^{(dir)}_\Lm: \sH^{(dir)}_\Lm \to \sH^{(dir)}_\Lm.
\label{eq:dirbc}
\end{equation}

Note that, since $\sH^{(dir)}_\Lm$ is finite-dimensional, $D^{(dir)}_{\uU,\Lm}$ can be represented by a matrix of size $(k |\Lm|) \times (k |\Lm|)$, where $|\Lm|$ denotes the number of elements in $\Lm$. Since $D_\uU$ is self-adjoint, so is $D^{(dir)}_{\uU,\Lm}$.
The number of eigenvalues of $D^{(dir)}_{\uU,\Lm}$ smaller than some $E\in\RR$, counting multiplicity, is denoted by  
\begin{equation}\label{def:nev}
	N^{(dir)}_{\Lm,\uU}(E) \ := \ \Tr \big\{ \One[D^{(dir)}_{\uU,\Lm}\ < E] \big\}. 
\end{equation}
The integrated density of states of $D^{(dir)}_{\uU,\Lm}$ is defined as the number of eigenvalues smaller than $E$ per unit volume,
\begin{equation}\label{def:ids}
\rho^{(dir)}_{\Lm,\uU}(E)\ := \ \frac{1}{|\Lm|} N^{(dir)}_{\Lm,\uU}(E).
\end{equation} 
Clearly, $N^{(dir)}_{\Lm,\uU}(E)$ depends only on the gauge fields on the bonds connecting points in $\Lm$.
Note that probabilistic statements about $N^{(dir)}_{T^\ell \Lm,\uU}(E)$ do not depend on $\ell \in \Zd$, since $\PP$ and $D_\uU$ are stationary.

%%%%%%%%%%%%%%%%%%%%%%%%%%%%%%%%%%%%%%%%%%%%%%%%%%%%%%
\subsubsection{Periodic Boundary Conditions}%%%%%%%%%%
%%%%%%%%%%%%%%%%%%%%%%%%%%%%%%%%%%%%%%%%%%%%%%%%%%%%%%
Another way to restrict $D_\uU$ to a finite set $\Lm\subset\Zd$ is to require periodic boundary conditions, which is often used in numerical simulations. 
In order to define periodic boundary conditions we assume $\Lm$ to be a cube of side length $L$. Without loss of generality we may assume that $\Lm=\{1,2,\ldots,L\}^d$.

We define the Wilson Dirac operator on $\ell^2(\Lm^{(per)};\CC^k)$ with periodic boun\-dary conditions by
\begin{equation}
	\big[ D_\uU^{(per)} \phi \big] (x)\ :=\ \gamma_5 \big[ \phi(x)- \kappa \sum_{\mu=1}^{4} \sum_{\sigma=\pm 1}  (r-\sigma \gamma_\mu) U_{x,\sigma\mu}\ \phi(x+\sigma\he_\mu)\big]\ ,
\end{equation}
where $\Lm^{(per)} := (\Z/L\Z)^d$ and $x+\sigma\he_\mu$ is determined only modulo multiples of $L$ in all directions. Similarly
\begin{equation}
	U_{x,-\mu} = U_{x-\he_\mu,\mu}^{-1},
\end{equation}
where $x-\he_\mu$ is also defined modulo $L$.
Thus $D_\uU^{(per)} : \sH_\Lm^{(per)} \to \sH_\Lm^{(per)}$, with
\begin{equation}
	\sH_\Lm^{(per)}\ :=\ \ell^2(\Lm^{(per)};\CC^k),
\end{equation}
only depends on the values of \uU for bonds $b \in \Lm\times \{1,\ldots,d\}$, i.e., on
\begin{equation}
	\uU_\Lm := \{ U_{x,\mu} \}_{x\in\Lm,\mu=1,\ldots,d}.
\end{equation}
Since $\sH_\Lm^{(per)}$ is finite-dimensional, we can transcribe the definition of the integrated density of states to periodic boundary conditions. 
We set $N_{\Lm,\uU}^{(per)}(E)$ to be the number of eigenvalues, counting multiplicity, of $D_{\uU,\Lm}^{(per)}$ smaller than $E \in \RR$ and define the integrated density of states in the periodic case as
\begin{equation}
	\rho_{\uU,\Lm}^{(per)}(E):=\inv{|\Lm|} N^{(per)}_{\Lm,\uU}(E).
\end{equation}

%\newpage

%%%%%%%%%%%%%%%%%%%%%%%%%%%%%%%%%%%%%
\section{Main Theorem}%%%%%%%%%%%%%%%
%%%%%%%%%%%%%%%%%%%%%%%%%%%%%%%%%%%%%

Our aim is the definition of the integrated density of states for $\{D_\uU\}_{\uU \in \Om}$.
A natural way is to let $\Lm$ be a cube of side length $L$ and investigate the case $L \to \infty$, the thermodynamic limit. As it turns out, the boundary conditions imposed is immaterial.

\begin{Theo}\label{th:mainids}
Let $(\cG^\cB,\cF,\PP)$ be the probability space defined in Section \ref{sec:introduction} and choose $0<\beta<\frac{1}{12N(d-1)}$ (such that $\PP$ is ergodic). Let $\{D_\uU\}_{\uU \in \cG^\cB}$ be the family of Wilson Dirac operators on $\sH$. 
Let $({\Om}_n)_{n\in\NN}$ be a sequence in $\Z^d$ of nested cubes, $\Om_n\subseteq\Om_{n+1}$, with $\Om_n\nearrow\Z^d$.
\begin{enumerate}
	\item 
Then the limits
\begin{equation}\label{eq:ergdir}
	\rho^{(dir)}_\uU(E)\ :=\ \lim_{n \to \infty} \frac{1}{|\Om_n|} N^{(dir)}_{\Om_n,\uU}(E)
\end{equation} 
and 
\begin{equation}\label{eq:ergper}
	\rho^{(per)}_\uU(E)\ :=\ \lim_{n \to \infty} \frac{1}{|\Om_n|} N^{(per)}_{\Om_n,\uU}(E)
\end{equation}
exist for all $E \in \RR$, $\PP$-almost surely, and are independent of the sequence $({\Om}_n)_{n\in\NN}$.
\item
Furthermore, for all $E \in \RR$, the \ul{integrated density of states} $\rho(E)$, defined by
\begin{equation}\label{eq:ids}
	\rho^{(dir)}_\uU(E) = \rho^{(per)}_\uU(E) =: \rho(E),
\end{equation}
is independent of the chosen boundary condition and of $\uU \in \cG^\cB$, $\PP$-almost surely.
\end{enumerate}
\end{Theo}

%%%%%%%%%%%%%%%%%%%%%%%%%%%%
\section{Proof}%%%%%%%%%%%%%
%%%%%%%%%%%%%%%%%%%%%%%%%%%%

%%%%%%%%%%%%%%%%%%%%%%%%%%%%%%%%%%%%%%%%%%%%%%%%%%%%%%%%%%%%%%%%%%
\subsection{An Estimate on Eigenvalues}\label{seq:evestimate}%%%%%
%%%%%%%%%%%%%%%%%%%%%%%%%%%%%%%%%%%%%%%%%%%%%%%%%%%%%%%%%%%%%%%%%%

Suppose, we take two disjoint sets $\Om_1,\Om_2\subset\Zd$ and an operator $D_\uU$, that fulfills the requirements of Theorem \ref{th:mainids}.
We can restrict $D_\uU$ to $\Om_1$, $\Om_2$ and $\Om_1 \cup \Om_2$ as in (\ref{eq:dirbc}) by means of the projections $P_{\Om_{1}}$, $P_{\Om_{2}}$ and $P_{\Om_{1}\cup\Om_{2}}$. 
Then we can determine the number of eigenvalues below some $E\in \RR$ for all three restrictions, denoted by $N_{\uU,\Om_1}(E)$, $N_{\uU,\Om_2}(E)$ and $N_{\uU,\Om_1 \cup \Om_2}(E)$, respectively.
If $\dist(\Om_1,\Om_2)\geq 2$, we know that 
$N_{\uU,\Om_1 \cup \Om_2}(E) = N_{\uU,\Om_1}(E) + N_{\uU,\Om_2}(E)$, 
since the operator $D_\uU$ links only neighbouring sites.
Our first goal is the derivation of an upper bound on the difference of $N_{\uU,\Om_1}(E) + N_{\uU,\Om_2}(E)$ and $N_{\uU, \Om_1 \cup \Om_2}(E)$.

To this end, we start with a general observation for finite matrices.

Let $A,B$ be complex, self-adjoint $(M \times M)$-matrices. The rank of $B$ is denoted by $b$, and the interesting case is $b \ll M$.
Since $A$ and $B$ are self-adjoint, so is $A+B$, and all three matrices $A,B$ and $A+B$ have $M$ real eigenvalues counting multiplicity.
Due to the fact that $\rk(B)=b$, $B$ has $(M-b)$ eigenvalues equal to zero, and $b$ eigenvalues different from zero. 
Furthermore, we denote by $N_A \in \NN_0$ the number of negative eigenvalues of $A$, by $N_B$, $N_{-B}$, $N_{A+B}$ the number of negative eigenvalues of $B$, $-B$, and $A+B$, respectively.

\begin{Lemma}\label{lm:eigenvalues}
Let A, B  be self-adjoint $M\times M$-matrices.
	The difference of the number $N_A$ of negative eigenvalues of $A$ and the number $N_{A+B}$ of negative eigenvalues of $A+B$ is at most $\rk(B)$, 
	\begin{equation}\label{eq:ev}
	|N_A-N_{A+B}| \leq \rk(B). 
	\end{equation}
\end{Lemma}

Note, that the bound (\ref{eq:ev}) is independent of $\left\|B \right\|$.

\begin{proof}
First, we show that $N_{A+B}-N_A\leq\rk(B)$.
Let us assume that $N_{A+B}>N_A+\rk(B)$.
Then the min-max principle ensures the existence of a subspace $X\subseteq\CC^M$, with dimension $\dim(X)=N_A+\rk(B)+1$ such that
\begin{equation}
	\sup_{\phi\in X, \|\phi\|=1} \langle \phi | (A+B) \phi \rangle < 0.
\end{equation} 
In particular we have 
\begin{equation}
	\sup_{\phi\in X \cap \ker(B), \|\phi\|=1} \langle \phi | (A+B) \phi \rangle 
	= \sup_{\phi\in X \cap \ker(B), \|\phi\|=1} \langle \phi | A \phi \rangle  < 0.
\end{equation} 
Using the min-max principle again, we obtain 
$N_A \geq \dim(X\cap\ker(B)) \geq \dim(X) - \rk(B)= N_A +1$. Therefore we have that $N_{A+B} - N_A \leq \rk(B)$.

Now, we set $A':=A+B$, $B':= B$ and get analogously $N_{A'+B'} - N_A' \leq \rk(B')$ that is $N_{A} - N_{A+B} \leq \rk(B)$.
\end{proof}

\begin{Lemma}\label{lm:idsdiff}
Let $(\cG^\cB,\cF,\PP)$ be the probability space defined in section \ref{sec:probspace} and choose $0<\beta<\frac{1}{12N(d-1)}$ such that $\PP$ is ergodic. Let $\{D_\uU\}_{\uU \in \cG^\cB}$ be the family of Wilson Dirac operators on $\sH$. 
Furthermore let $\Om_1,\ldots, \Om_J \subset \Z^d$ be disjoint, finite sets and $\Om:= \bigcup_{j=1}^{J} \Om_j$ their union.
\begin{enumerate}
\item
Then we have, for any $\uU \in \cG^\cB$ and all $E \in \RR$,
\begin{equation}\label{eq:idsdiffsim}
	\inv{|\Om|} \left| N^{(dir)}_{\uU,\Om}(E) - \sum_{j=1}^{J} N^{(dir)}_{\uU,\Om_{j}}(E)\right|\ \leq\ 
	k\ \frac{\sum_{j=1}^{J} |\dd\Om_j|}{ |\Om|}.
\end{equation}
\item 
If in addition the sets $\Om_1,\ldots,\Om_J$ are cubes, such that $\Om$ is also a cube, then 
\begin{equation}\label{eq:idsdiffper}
	\inv{|\Om|} \left| N^{(per)}_{\uU,\Om}(E) - \sum_{j=1}^{J} N^{(per)}_{\uU,\Om_{j}}(E) \right|\ \leq\ 
	3 k\ \frac{\sum_{j=1}^{J} |\dd\Om_j|}{ |\Om|},
\end{equation}
for any $\uU \in \cG^\cB$ and all $E \in \RR$.
\end{enumerate}
\end{Lemma}

\begin{proof}
We remark that it is enough to prove the case $E=0$, because we can replace $D_\uU$ by $D_\uU-E$.
We start with the case of Dirichlet boundary conditions and define $(k|\Om| \times k|\Om|)$-matrices $A$ and $C$ by
\begin{equation}\label{def:matrixA}
	A := D^{(dir)}_{\uU,\Om},\ C:=\sum_{j=1}^J P_{\Om_j} A P_{\Om_j} \ \ :\ \sH^{(dir)}_\Om \to \sH^{(dir)}_\Om.
\end{equation}
The matrices are chosen in such a way that we get with counting multiplicity
\begin{equation}
	N^{(dir)}_{\uU,\Om}(0)= \Tr \big\{ \One[A \geq 0] \big\} 
\end{equation}
and
\begin{equation}
	\sum_{j=1}^{J} N^{(dir)}_{\uU,\Om_{j}}(0)=\Tr \big\{ \One[C \geq 0] \big\}.
\end{equation}
Now, we set the matrix $B:=A-C$ to be the difference of $A$ and $C$. 
The rank of $B$ can be estimated as follows
\begin{equation}
	\rk(B) \leq k \sum_{j=1}^J |\dd\Om_j|,
\end{equation}
using that $B=\sum_{j \neq l} P_{\Om_j} A P_{\Om_l}$.
By Lemma~\ref{lm:eigenvalues}, we obtain
\begin{equation}
	\inv{|\Om|} \left| N_{\uU,\Om}(0) - \sum_{j=1}^{J} N_{\uU,\Om_{j}}(0) \right| 
	\leq k\ \frac{\sum_{j=1}^{J} |\dd\Om_j| }{ |\Om|}\ ,
\end{equation} 
and \textit{(i)} is proven.

To prove \textit{(ii)}, we use that, for any cube $\Lm$, we have that 
\begin{equation}
	\rk\big[ D^{(per)}_{\uU,\Lm} - D^{(dir)}_{\uU,\Lm}  \big]\ \leq\ k \ |\dd\Lm|
\end{equation}
Therefore, \textit{(i)} and another application of Lemma~\ref{lm:eigenvalues} yield (\ref{eq:idsdiffper}),
\begin{eqnarray}
	\lefteqn{\inv{|\Om|} \left| N^{(per)}_{\uU,\Om}(E) - \sum_{j=1}^{J} N^{(per)}_{\uU,\Om_{j}}(E) \right|
	 \leq  \inv{|\Om|} \Big( \left| N^{(per)}_{\uU,\Om}(E) - N^{(dir)}_{\uU,\Om}(E) \right| }\nonumber\\
	& +&  \left| N^{(dir)}_{\uU,\Om}(E) - \sum_{j=1}^{J} N^{(dir)}_{\uU,\Om_{j}}(E) \right|
		+  \left| \sum_{j=1}^{J} (N^{(dir)}_{\uU,\Om_{j}}(E) - N^{(per)}_{\uU,\Om_{j}}(E)) \right| \Big) \nonumber\\
	& \leq & 3 k\ \frac{\sum_{j=1}^{J} |\dd\Om_j|}{ |\Om|}.
\end{eqnarray}

\end{proof}

Lemma~\ref{lm:idsdiff} is an estimate on the change of the integrated density of states as the subset of $\Zd$ is broken up into smaller pieces. The estimate is, indeed, precise enough to prove the existence of a limit in the sense of Theorem~\ref{th:mainids} as is done in the next sections.

%%%%%%%%%%%%%%%%%%%%%%%%%%%%%%%%%%%%%%%%%%%%%%%%%%%%%%%%%%%%%%%%%%%%%%%%%%%%%%%%%%%%%%%%%%%%
\subsection{Existence of the Integrated Density of States for a Special Sequence}%%%%%%%%%%%
%%%%%%%%%%%%%%%%%%%%%%%%%%%%%%%%%%%%%%%%%%%%%%%%%%%%%%%%%%%%%%%%%%%%%%%%%%%%%%%%%%%%%%%%%%%%

In this section it is shown that a limit for the integrated density of states exists almost surely for a sequence of growing cubes in $\Zd$. 

To this end, we define the following sequence of growing cubes,
\begin{equation}\label{def:seq}
	\Lm_n := \{ - l_0 2^{n-1} +1, \ldots, l_0 2^{n-1} \}^d \ ,
\end{equation}
with $l_0 \in \NN$ to be fixed later. Note that $\Lm_n$ has side length $l_0 2^n$.

The virtue of the sequence $(\Lm_n)_{n \in \NN}$ is, that $\Lm_{n+1}$ splits into $2^d$ disjoint cubes, each of size $|\Lm_{n}|$, in a natural way.
More precisely, there are $z_1,\ldots,z_{2^d} \in \Z^d$ such that, for 
\begin{equation}\label{def:sets}
	\Pi_n:=\{T^{z_1}, \ldots, T^{z_{2^d}} \}
\end{equation}
being the set of associated translations,
\begin{equation}\label{def:split}
	\Lm_{n+1}=\bigcup_{T \in \Pi_n} T \Lm_n.
\end{equation}

In order to clarify the notation, we also introduce the sets $\Pi^l_n$ for $l>n$ that consist of the translations needed to compose $\Lm_l$ of translations of $\Lm_n$,
\begin{equation}
 \Pi^l_n:= \big\{ T_n \ldots T_{l-1} : T_n \in \Pi_n, \ldots, T_{l-1} \in \Pi_{l-1} \big\}.
\end{equation}
Thus, $\Pi_n=\Pi^{n+1}_n$ and we have
\begin{equation}\label{def:split2}
	\Lm_{l}=\bigcup_{T \in \Pi^l_n} T \Lm_n,
\end{equation}
see Figure~\ref{fig:lnlm}.
\begin{figure}
	\begin{center}
		\includegraphics[width=0.50\textwidth]{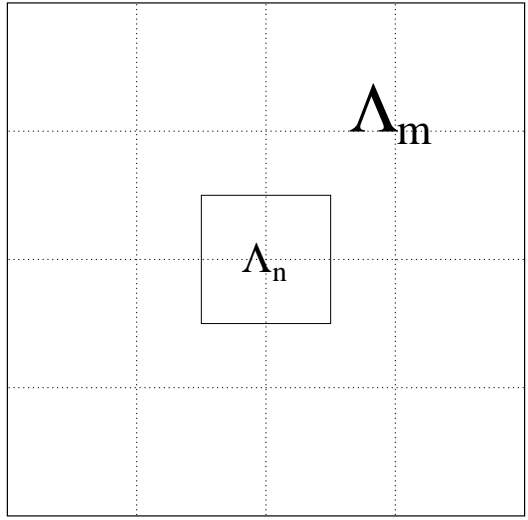}
	\end{center}
	\caption{$\Lm_m$ and $\Lm_n$ in $\Z^2$, with $m=n+2$. 
					 The dotted lines indicate the translates of $\Lm_n$ whose union gives $\Lm_m$.}
	\label{fig:lnlm}
\end{figure}
Next, we study the integrated density of states of $\Lm_{n}$, as $n$ grows. 
We omit the dependence of 
$\displaystyle N^{(dir)}_{\Lm_{n},\uU}(E)$ and 
$\displaystyle N^{(per)}_{\Lm_{n},\uU}(E)$ on $E$ and the gauge field configuration $\uU$ and write 
\begin{equation}
	N^{(dir)}[\Lm_n]\ :=\ N^{(dir)}_{\Lm_{n},\uU}(E) 
	\quad\text{and}\quad 
	N^{(per)}[\Lm_n]\ :=\ N^{(per)}_{\Lm_{n},\uU}(E)
\end{equation}
instead.

\begin{Lemma}\label{lm:aux1}
For any $l_0 \in \NN$, the sequences 
$\big( \inv{|\Lm_n|}  N^{(dir)}[\Lm_n] \big)_{n\in \NN}$ and \newline
$\big( \inv{|\Lm_n|}  N^{(per)}[\Lm_n] \big)_{n\in \NN}$ converge, $\PP$-almost surely.
\begin{equation}
	\lim_{n\to\infty} \inv{|\Lm_n|}  N^{(dir)}[\Lm_n] \to \rho^{(dir)}_{l_0}, \qquad
	\lim_{n\to\infty} \inv{|\Lm_n|}  N^{(per)}[\Lm_n] \to \rho^{(per)}_{l_0}.
\end{equation}
Furthermore
\begin{equation}\label{eq:pereq1}
	\lim_{n \to \infty} \inv{|\Lm_n|} \left| N^{(dir)}[\Lm_n] - N^{(per)}[\Lm_n] \right| = 0 ,
\end{equation}
$\PP$-almost surely, and $\rho_{l_0}:=\rho^{(dir)}_{l_0}=\rho^{(per)}_{l_0}$.
\end{Lemma}

\begin{proof}
We show that $\big(\inv{|\Lm_n|} N^{(dir)}[\Lm_n] \big)_{n\in\NN}$ is a Cauchy sequence and in the proof we denote $N^{(dir)}[\Lm_n]=:N[\Lm_n]$.
The proof for periodic boundary conditions is completely analogous.

Assume that $m > n$. 
By applying (\ref{def:split2}), one can split $\Lm_m$ into $2^{d(m-n)}$ cubes of size $|\Lm_n|$,
\begin{equation}
	\Lm_{m} = \bigcup_{T \in \Pi^m_n} T \Lm_n.
\end{equation} 
The mean integrated density of states for these translations of $\Lm_n$ is
\begin{equation}
	\inv{ 2^{d(m-n)} } \sum_{T \in \Pi^m_n} \inv{|\Lm_n|}  N[T \Lm_n]
	\	=\ \inv{|\Lm_m|} \sum_{T \in \Pi^m_n}  N[T \Lm_n].
\end{equation}
Thus we can estimate  
\begin{eqnarray}\label{eq:imp1}
\lefteqn{
	\left| \inv{|\Lm_m|}  N[\Lm_m]  - \inv{|\Lm_n|}  N[\Lm_n]   \right|} \\\nonumber
	&\leq &
	\inv{|\Lm_m|}	\left|	N[\Lm_m] - \sum_{T \in \Pi^m_n} N[T \Lm_n] 	\right| 
	+ 	\left|	\inv{|\Lm_m|} \sum_{T \in \Pi^m_n}  N[T \Lm_n] - \inv{|\Lm_n|}  N[\Lm_n] 	\right|	.
\end{eqnarray}
Lemma~\ref{lm:idsdiff} directly gives us an upper bound for the first term on the right side of (\ref{eq:imp1}), since we have
\begin{eqnarray}\label{eq:imp2}
	\lefteqn{
	\inv{|\Lm_m|} \left|	N[\Lm_m] - \sum_{T \in \Pi^m_n} N[T \Lm_n] \right|
	\ \leq\ 	k\ \frac{ 2^{d(m-n)} |\dd \Lm_n|}{|\Lm_m|} }\nonumber \\
	&\leq&  
	k\ 2^{d(m-n)}\ \frac{  2 d (l_0 2^n)^{d-1}}{(l_0 2^m)^d}
	\ =\  \frac{2dk}{l_0}\ 2^{-n},
\end{eqnarray}
independently of the gauge field configuration.

The second term on the right side of (\ref{eq:imp1}) is the difference of the integrated density of states for a cube $\Lm_n$ and its spatial mean over $2^{d(m-n)}$ translated disjoint cubes of the same size. 
As we do not know, yet, whether this term is small with high probability, provided $n$ is large enough,
we split $\Lm_n$ and its translates into smaller cubes of size $|\Lm_{n_0}|$ for some $n_0<n \in \NN$, as indicated in Figure~\ref{fig:lnlm2}.

\begin{figure}
	\begin{center}
		\includegraphics[width=0.50\textwidth]{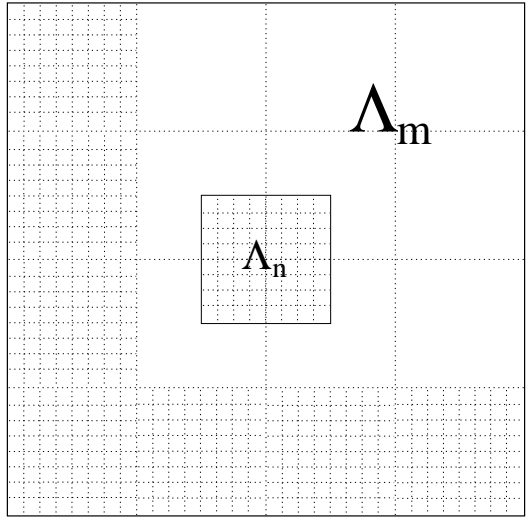}
	\end{center}
	\caption{Both sets $\Lm_n$ and $\Lm_m$ are split up in smaller cubes of the same size as $\Lm_{n_0}$. For $\Lm_m$ only part of the splitting is sketched.}
	\label{fig:lnlm2}
\end{figure}

We estimate
\begin{eqnarray}\label{eq:imp3}
	\lefteqn{
	\left|
	\inv{|\Lm_m|} \sum_{T \in \Pi^m_n} N[T \Lm_n] 
	- \inv{|\Lm_n|} N[\Lm_n]
	\right|	} \nonumber \\
 &\leq	 &
  \inv{|\Lm_m|} \sum_{T \in \Pi^m_n} 	
  \left|
  N[T\ \Lm_n] - 
  \sum_{T' \in \Pi^n_{n_0}} N[T' T \Lm_{n_0}]
	\right|	\nonumber \\
 &&+ 
 	\left|
	\inv{|\Lm_m|} \sum_{T \in \Pi^m_{n_0}} N[T \Lm_{n_0}]
	-	\inv{|\Lm_n|} \sum_{T \in \Pi^n_{n_0}} N[T \Lm_{n_0}]
	\right|	\nonumber \\
 && {} + 
 	\inv{|\Lm_n|}
	\left|
  \sum_{T \in \Pi^n_{n_0}} N[T \Lm_{n_0}]
  - N[\Lm_{n}]
	\right|	,
\end{eqnarray}
using that any $T \in \Pi^m_{n_0}$ is given as a product $T=T' T''$, for unique $T' \in \Pi^m_n$ and $T'' \in \Pi^n_{n_0}$.
Then Lemma~\ref{lm:idsdiff} yields again an upper bound for the first and the third term on the right side of (\ref{eq:imp3}), and analogously to (\ref{eq:imp2}), we obtain that
\begin{equation}\label{eq:imp4}
	\inv{|\Lm_m|} \sum_{T \in \Pi^m_n} 	\left| N[T \Lm_n] - \sum_{T' \in \Pi^n_{n_0}} N[T' T \Lm_{n_0}] \right|\ 
	\leq\ \frac{2dk}{l_0}\ 2^{-n_0}
\end{equation}
and
\begin{equation}\label{eq:imp4a}
	\inv{|\Lm_n|}	\left| \sum_{T \in \Pi^n_{n_0}} N[T \Lm_{n_0}] - N[\Lm_{n}]	\right|\
	\leq\ \frac{2dk}{l_0}\ 2^{-n_0}. 
\end{equation} 
Thus Equations \eqref{eq:imp2}, \eqref{eq:imp3}, \eqref{eq:imp4}, and \eqref{eq:imp4a} yield
\begin{eqnarray}\label{eq:imp5}
\lefteqn{
	\left|
	\inv{|\Lm_m|} N[\Lm_m] - \inv{|\Lm_n|} N[\Lm_n]
	\right|} \\\nonumber
& \leq &
	\frac{4dk}{l_0} (2^{-n}	+	2^{-n_0})
	+ \bigg|
	\inv{|\Lm_m|} \sum_{T \in \Pi^m_{n_0}} N[T \Lm_{n_0}]
 	-	\inv{|\Lm_n|} \sum_{T \in \Pi^n_{n_0}} N[T \Lm_{n_0}]	
 	\bigg|	.
\end{eqnarray}
We can choose $n_0$ and then $n>n_0$  so large that $\frac{4dk}{l_0} (2^{-n}	+	2^{-n_0})$ is arbitrarily small.
To estimate the remaining term, we view $\{ Z_x(\uU) \}_{x\in\Zd}$, with 
\begin{equation}
	Z_x(\uU):=N[T^{2^{n_0} l_0 x } \Lm_{n_0}]=N^{(dir)}_{T^{2^{n_0} l_0 x}\Lm_{n_0},\uU}(E),
\end{equation}
to be an invariant family of random variables. By Birkhoff's Ergodic Theorem, the mean of these random variables converges $\PP$-almost surely. Hence,
\begin{equation}
	\bigg( \inv{|\Lm_m|} \sum_{T \in \Pi^m_{n_0}} N[T \Lm_{n_0}] \bigg)_{m=n_0+1}^\infty
\end{equation}
is a Cauchy sequence, $\PP$-almost surely.

As noted above, we can replace $N[\#]$ by $N^{(per)}[\#]$ and repeat the proof for periodic boundary conditions with exactly the same arguments, since all sets are cubes and Lemma~\ref{lm:idsdiff}(ii) applies.

Equation~(\ref{eq:pereq1}) is similarly proven as Lemma~\ref{lm:idsdiff}(ii), 
\begin{equation}
	\inv{|\Lm_n|} \left| N^{(dir)}[\Lm_n] - N^{(per)}[\Lm_n] \right| 
	\ \leq \ 
	\frac{k |\dd\Lm_n|}{|\Lm_n|} \ \to\ 0,\quad n\to\infty\ .
\end{equation}
\end{proof}

Note that, while the preceding lemma holds for all $l_0\in\NN_0$, this does not imply the independence of the integrated density of states of the choice of $l_0$. It turns out, however, that not only the independence holds true, but that furthermore the size of the cubes in the sequence is immaterial, as long as it is monotonically growing.

%%%%%%%%%%%%%%%%%%%%%%%%%%%%%%%%%%%%%%%%%%%%%%%%%%%%%%
\subsection{Proof of Main Theorem~\ref{th:mainids}}%%%%%%%%
%%%%%%%%%%%%%%%%%%%%%%%%%%%%%%%%%%%%%%%%%%%%%%%%%%%%%%

The proof is similar to the one of Lemma~\ref{lm:aux1}.
We choose $l_0, n_0 \in \NN$ arbitrary, but fixed. Given $\Om_k$, there is an $m\in\NN$ such that $\Om_k \subseteq \Lm_m$. We define
\begin{eqnarray}	
	\Sigma_k &:=& \big\{ T \in \Pi^m_{n_0} \ |\ T\Lm_{n_0}\subseteq \Om_k  \big\}, \nonumber\\
	\tOm_k   &:=& \bigcup_{T\in\Sigma_k} T \Lm_{n_0} \ \subseteq\ \Om_k
\end{eqnarray}

Note that $\tOm_k$ is a rectangular box, whose smallest side length is at most two times smaller than its largest side length and all side lengths are multiples of $l_0 2^{n_0}$. Moreover
\begin{equation}
	|\Om_k|-|\tOm_k| = |\Om_k\wo\tOm_k| \leq |\dd\Om_k| \cdot |\Lm_{n_0}|
\end{equation}

Now, we estimate for $k>n_0$, $\Om_k \supseteq \Lm_{n_0}$
\begin{eqnarray}\label{eq:imp6b}
	\lefteqn{	\left| \inv{|\Lm_k|} N[\Lm_k] - \inv{|\Om_k|} N[\Om_k] \right|
	\ \leq\ 
	\left| \inv{|\Lm_k|} N[\Lm_k] -  \inv{|\Lm_k|} \sum_{T \in \Pi^k_{n_0}} N[T \Lm_{n_0}] \right| 
	}	\nonumber\\ && {}	+
	\left| \inv{|\Lm_k|} \sum_{T \in \Pi^k_{n_0}} N[T \Lm_{n_0}] - \inv{|\tOm_k|} \sum_{T \in \Sigma_k} N[T \Lm_{n_0}] \right| 
	\nonumber\\	&& {}	+ 
	\left| \inv{|\tOm_k|} \sum_{T \in \Sigma_k} N[T \Lm_{n_0}] -  \inv{|\tOm_k|} N[\tOm_k]  \right|
	\nonumber\\&& {}	+ 	
	\left|  \inv{|\tOm_k|} N[\tOm_k] - \inv{|\Om_k|}  N[\Om_k]\right|.
\end{eqnarray}
In the first and third term we can apply Lemma~\ref{lm:idsdiff} directly to get upper bounds. The second term converges $\PP$-almost surely to zero, by Birkhoff's ergodic theorem.
For the last term we estimate the difference of the integrated density of states of $\tOm_k$ and $\Om_k$,
\begin{equation}\label{eq:diffids}
	\left| \inv{|\tOm_k|} N[\tOm_k] - \inv{|\Om_k|} N[\Om_k] \right|
	=\frac{\left| |\tOm_k|\cdot N[\Om_k] - |\Om_k| \cdot N[\tOm_k] \right|}{ |\tOm_k| \cdot |\Om_k| } .
\end{equation}
First, we estimate
\begin{eqnarray}	\label{eq:est2}
\lefteqn{ \bigg| |\Om_k| \cdot N[\tOm_k] - |\tOm_k| \cdot N[\Om_k] \bigg|} \nonumber\\
		&=& \bigg| |\Om_k| \big( N[\tOm_k] - N[\Om_k]\big) + \big(|\Om_k| - |\tOm_k|\big) N[\Om_k] \bigg| \nonumber\\
	  &\leq& |\Om_k| \big|N[\Om_k] - N[\tOm_k]\big| + \big(|\Om_k|-|\tOm_k|\big) N[\Om_k].
\end{eqnarray}
Then we observe that
\begin{equation}\label{eq:diff2}
	|N[\Om_k]-N[\tOm_k]|
	\ \leq\ 
	k|\dd\tOm_k| + (2d+1)k(|\Om_k|-|\tOm_k|)
\end{equation}
holds true because of the following estimate using Lemma~\ref{lm:idsdiff}
\begin{eqnarray}
\lefteqn{ \left| N[\Om_k] - N[\tOm_k] \right| - \left| \sum_{x\in\Om_k\wo\tOm_k} N[\{x\}] \right|} \nonumber\\
	& \leq & \left| N[\Om_k] - \bigg( N[\tOm_k] + \sum_{x\in\Om_k\wo\tOm_k} N[\{x\}] \bigg) \right| \nonumber\\
	& \leq & k\Big( |\dd\tOm_k| + 2d(|\Om_k|-|\tOm_k|) \Big).
\end{eqnarray}
Thus we get
\begin{eqnarray}
\lefteqn{\left| \inv{|\tOm_k|} N[\tOm_k] - \inv{|\Om_k|} N[\Om_k] \right|} \nonumber\\
	&\leq & \inv{|\tOm_k|} |N[\Om_k] - N[\tOm_k]| + \inv{|\Om_k|}N[\Om_k]\frac{(|\Om_k|-|\tOm_k|)}{|\tOm_k|} \nonumber\\
	&\leq & \frac{k|\dd\tOm_k|}{|\tOm_k|} 
	+ \frac{(2d+2)k(|\Om_k|-|\tOm_k|)}{|\tOm_k|}
	\ \stackrel{n\to\infty}{\longrightarrow}\ 0, 
\end{eqnarray}
by using Equations \eqref{eq:est2}, \eqref{eq:diff2} and the fact that $\inv{|\Om_k|}N[\Om_k]\leq k$.
Altogether we have proven
\begin{equation}\label{eq:result}
	\lim_{k\to\infty} 	\left| \inv{|\Lm_k|} N[\Lm_k] - \inv{|\Om_k|} N[\Om_k] \right| = 0,
\end{equation}
$\PP$-almost surely.
The periodic case is again proven completely analogously.

Recall that Lemma~\ref{lm:aux1} gives the  existence of the limit 
\begin{equation}
	\rho_\uU (E) := \lim_{n \to \infty}\frac{1}{|\Lm_n|} N^{(dir)}_{\Lm_n,\uU}(E)
	=\lim_{n \to \infty}\frac{1}{|\Lm_n|} N^{(per)}_{\Lm_n,\uU}(E)
\end{equation}
with ${(\Lm_n)}_{n \in \NN}$ as in \eqref{def:seq}, $\PP$-almost surely.
Equation \eqref{eq:result} ensures the independence of this limit of the chosen sequence $(\Om_k)_{k\in\NN}$.
Furthermore, Lemma~\ref{lm:aux1} implies that the limit for periodic boundary conditions is the same.

Since the integrated density of states is invariant under translations, it is $\PP$-almost surely constant, and \textit{(ii)} follows.
This finishes the proof of Theorem \ref{th:mainids}. \hfill \qed

\section*{Acknowledgements}
C.K. gratefully acknowledges generous support by the Graduate School of Excellence \textquoteleft Elementare Kr\"afte und mathematische Grundlagen EMG\textquoteright. The authors thank H.~Wittig and T.~Wettig for drawing their attention to the link between lattice QCD and RMT.
The authors are further indebted to H.~Wittig for various suggestions that considerably improved this paper.

\bibliography{bibliography}{}
\bibliographystyle{plain}

\end{document}